\documentclass[12pt]{amsart}
\usepackage{amssymb, amsmath}
\usepackage{graphics}
\usepackage{epsfig}
\usepackage{amsxtra}
\usepackage{color}

\numberwithin{equation}{section}

\newtheorem{thm}{Theorem}[section]

\newtheorem{lemma}[thm]{Lemma}

\newtheorem{prop}[thm]{Proposition}
\newtheorem{cor}[thm]{Corollary}

\theoremstyle{definition}

\newcommand{\defeq}{\stackrel{\rm{def}}{=}}

\renewcommand{\Re}{\operatorname{\rm Re}\nolimits}
\renewcommand{\Im}{\operatorname{\rm Im}\nolimits}

\def \rank {\operatorname{rank}}
\def \comp {\operatorname{comp}}
\def \supp {\operatorname{supp}}

\def \supp {\operatorname{supp}}

\def \mco {{\mathcal O}}
\def \mca {{\mathcal A}}

\def \vol {{\rm vol}}

\def \restrict {\upharpoonright}

\def \mch {{\mathcal H}}
\def \Real {{\mathbb R}}

\def \Sphere {\mathbb{S}}
\def \Complex {\mathbb{C}}
\def \Natural {{\mathbb N}}
\def \Sphere {{\mathbb S}}

\def \mcr{{\mathcal R}}

\def \U+ {U_+}
\def \Integers {{\mathbb Z}}

\title [Lower bounds for resonance counting functions]
{Lower bounds for resonance counting functions for 
obstacle scattering in even dimensions
}
   \author { T.J. Christiansen}
\address{Department of Mathematics,
University of Missouri,
Columbia, Missouri 65211, USA} 
\email{christiansent@missouri.edu}

\begin{document}

\begin{abstract}  In even dimensional Euclidean scattering, the 
resonances lie on  the logarithmic cover of the complex plane.
This paper studies resonances for obstacle scattering in $\Real^d$
with Dirchlet
or admissable Robin boundary conditions, when $d$ is even.
Set $n_m(r)$ to be the
number of resonances with norm at most $r$ and argument between $m\pi$ and
$(m+1)\pi$.  Then  $\lim\sup _{r\rightarrow \infty}\frac{\log n_m(r)}{\log r}=d$
if $m\in \Integers \setminus \{ 0\}$.
\end{abstract}

\maketitle

\section{Introduction}
This paper studies resonances for scattering by obstacles in  even
dimensions. In this setting the 
resonances lie on the logarithmic cover of the complex plane.
 The main result is that for $m\in \Integers \setminus \{0\}$
the  counting function for the number of 
resonances with norm at most $r$ and argument between 
$m\pi$ and $(m+1)\pi$ has maximal order of growth.  To the
best of our knowledge, the only specific obstacles for which this has been 
known before are balls.  

Let $\mco\subset \Real^d$ be a bounded open set with smooth boundary $\partial \mco$,
and suppose $\Real^d \setminus \mco$ is connected.  When $\mco$ satisfies these conditions
we shall call it an {\em obstacle}.  Consider $-\Delta$ on $\Real ^d \setminus {\overline{\mco}}$,
where $\Delta \leq 0$ is the usual Euclidean Laplacian.  We impose either Dirichlet 
($u \restrict_{\partial (\Real^d \setminus \overline{\mco})}= 0$) or
Robin-type boundary conditions:
$$h(x) u (x)+\frac{\partial}{\partial n}u(x) =0 \; \text{on}\;\partial (\Real^d \setminus 
\overline{\mco})$$
where $n$ is the outward pointing unit normal to $\partial(\Real^d \setminus \overline{\mco})$
and $h\in C^1(\partial \mco)$ satisfies $h\geq 0$ in order to be admissable.
We note that by choosing $h\equiv 0$ we obtain the Neumann boundary condition.
We shall denote the corresponding operator 
(satisfying either the Dirichlet or admissable Robin-type boundary
condition) by $P$.
We choose the upper half-plane to be the physical half plane, so that 
$R(\lambda)=(P-\lambda^2)^{-1}$
is bounded on $L^2(\Real^d \setminus \overline{\mco})$ 
for $\Im \lambda >0$.  It is well known that for any 
$\chi \in L^{\infty}_{\comp}(\Real^d \setminus \overline{\mco})$, 
$\chi R(\lambda)\chi$ has a meromorphic continuation to $\Complex$
if $d$ is odd and to $\Lambda$, the logarithmic cover of $\Complex\setminus \{ 0\}$, if 
$d$ is even (e.g. \cite{sj-zw}).
  If $\chi=1$ in a neighborhood of $\overline{\mco}$ 
the location of the poles of $\chi R(\lambda)\chi$ 
is independent of the choice of such $\chi$.
The poles of this meromorphic continuation are called {\em resonances}.

We can describe a point $\lambda \in \Lambda$ by its modulus $|\lambda|$ and its argument $\arg \lambda$.  On
$\Lambda$ we do not identify points whose arguments differ by an integral multiple of $2\pi$.
We define, for $m\in \Integers$, 
$$\Lambda_m =\{ \lambda \in \Lambda: \; m\pi <\arg \lambda <(m+1)\pi\}$$
and call this the $m$th sheet of $\Lambda$.
We note that our choice of the physical half plane means that it is identified with
$\Lambda_0$ in the case of even $d$.

For $m\in \Integers$ and even $d$ we define the $m$th resonance counting function:
$$n_m(r)= \# \{ \text{ $\lambda_j$: $\lambda_j$ is a pole of 
$R(\lambda)$ with $m\pi <\arg \lambda_j <(m+1)\pi$, 
and $|\lambda_j|<r$}\}.$$
Here and everywhere we count with multiplicity, see (\ref{eq:Rmult}).

Then our main result is
\begin{thm}\label{thm:main}
Let $d$ be even, $\mco \subset \Real^d$ be
an obstacle as defined above, and suppose $\mco \not = \emptyset.$  Consider $-\Delta$
on $\Real^d \setminus \overline{\mco}$ with any of Dirichlet, Neumann, or admissable Robin-type 
boundary conditions. 
Then, with $n_m(r)$ the resonance-counting function for the $m$th sheet as defined above,
$$\lim \sup _{r\rightarrow \infty} \frac{ \log n_m(r)}{\log r} = d$$
for any $m\in \Integers \setminus \{0\}$.
\end{thm}
The quantity $\lim \sup_{r\rightarrow \infty}\frac{\log n_m(r)}{\log r}$ is
called the order of $n_m$.
By results of Vodev \cite{vodeveven, vodev2}, $d$ is the maximum possible
value of 
the order of $n_m$.  A result of Sj\"ostrand and Zworski 
\cite{sj-zwconvex}
for asymptotics of resonances in certain regions of $\Lambda_{\pm 1}$ for
the Dirichlet problem for convex obstacles has as a corollary that 
the order of $n_{\pm 1}(r)$ is at least $d-1$ for such cases.  

In general  lower bounds on resonance counting functions have proved
elusive.
We may contrast Theorem \ref{thm:main} with what is known for obstacle 
scattering in odd dimensions $d$, $d\geq 3$.  In odd dimensions the 
resonances lie on $\Complex$ as the double cover of $\Complex$.  The 
analogous counting function $n(r)$
is for the number of resonances with norm at most 
$r$, and there is a an upper bound of the form $Cr^d$ for sufficiently large 
$r$, with constant $C$ depending on the obstacle \cite{melroseobstacle}.  
In {\em odd} dimensions as far as we know the only specific obstacles for 
which it is known that $\lim \sup_{r\rightarrow \infty}(\log n(r)/\log r) =d$
are balls, for which stronger results are known \cite{stefanov, zwradpot}.  
However, from 
 \cite{bbmeromorphic} it is known that this limit must be $d$ for 
many star-shaped obstacles.

This paper also contains, for completeness,
 some basic results on the behavior of the 
scattering matrix $S(\lambda)$ at $0$ in even dimensions.  In particular, for 
a large class of operators (``black box'' compactly 
supported perturbations of the Laplacian) in {\em even} dimension $d$
 it is shown that $\lim_{\lambda \downarrow 0}
S(\lambda)=I$.  See Section \ref{s:smat0} for results, references, and 
further remarks.

Let $P$ be the Laplacian in the exterior of an obstacle with Dirichlet
or admissable Robin boundary conditions.
Let $S(\lambda)$ denote the scattering matrix associated with the 
operator $P$; its definition
is recalled in Section \ref{s:lowerbounds}. 
In both even and odd dimensions, when $\arg \lambda = \pi/2$
it is easier
to say something about the eigenvalues of $S(\lambda)-I$ than
it is for an arbitrary value of $\lambda $ with $0<\arg \lambda <\pi$.  
Lax and Philips \cite{la-ph},
Beale \cite{beale}, and Vasy \cite{vasy} used this for {\em odd} $d\geq 3$
 to obtain lower bounds of the type $c_0r^{d-1}$ for 
the number of pure imaginary resonances with norm 
at most $r$  for obstacle scattering
(\cite{beale,la-ph}) 
and for scattering by fixed-sign potentials (\cite{la-ph,vasy}).  
The situation is quite different
for {\em even} dimensions $d$.  
For 
obstacle scattering of the type we consider here or for scattering by fixed-sign potentials for any $m\in \Integers$
there are at most finitely many resonances with argument $\pi m + \pi/2$ 
 \cite{beale,ch-hi4}.

Here we make use of the behavior of $S(e^{i\pi /2}\sigma)-I$, $\sigma>0$, to prove our theorem.


We now outline the proof of our main theorem, which requires introducing some
notation and results from other works.
  With $R(\lambda)$ denoting the meromorphic
continuation of the resolvent $(P-\lambda^2)^{-1}$, 
for $\lambda_0\in \Lambda$, the
multiplicity of a pole of $R$ at $\lambda_0$ is defined to be
\begin{equation}\label{eq:Rmult}
\mu_R(\lambda_0)\defeq \rank \int_{\gamma_{\lambda_0}} R(\lambda) d\lambda
\end{equation}
where $\gamma_{\lambda_0}$ is a small posivitely oriented curve enclosing 
$\lambda_0$ and no poles of the resolvent, except, possibly, at $\lambda_0$.
For a scalar meromorphic function $f$ defined on $\Lambda$, $\lambda_0\in 
\Lambda$, we define
$m_{\text{sc}}(\lambda_0)=k\in \Integers$ if and only if $f(\lambda)(\lambda-
\lambda_0)^{-k}$ is bounded in a sufficiently small neighborhood of $\lambda_0$
and $\lim_{\lambda \rightarrow \lambda_0 } \left( f(\lambda) (\lambda -\lambda_0)^{-k} \right) \not = 0$.  Thus $m_{\text{sc}}$ is positive at zeros and 
negative at poles, and $m_{\text{sc}}(f,\lambda_0)=0$ if $\lambda_0$ is
neither a zero nor a pole of $f$.

\begin{prop}\label{p:chi-hi4} \cite[Corollary 4.9]{ch-hi4}
For $P$ as above, $m\in \Natural$, and 
$\lambda_1\in \Lambda$, 
$$\mu_R(\lambda_1 e^{im\pi})-\mu_R(\lambda_1) = 
m_{\text{sc}}(\det(mS(\lambda)-(m-1)I),\lambda_1).$$
\end{prop}
Analogs of this result are well known in odd dimensions and for
even dimensions  for $m=1$ and 
a limited subset of $\Lambda$; see, for example, 
\cite{jensen, nedelec, pe-zwsc, s-t}.

The restrictions we have put on the operator $P$ mean that neither $R(\lambda)$
nor the scattering matrix has any poles in the physical region $\Lambda_0$.
Thus we see from Proposition \ref{p:chi-hi4}
that to study the resonances on $\Lambda_m$, $m\in \Natural$,
it will suffice to study the zeros of the scalar function 
\begin{equation}
\label{eq:fm}
f_m(\lambda) \defeq 
\det(mS(\lambda)-(m-1)I)
\end{equation} on the physical sheet $\Lambda_0$, which we shall
identify with the upper half plane of $\Complex$.  The function 
$f_m$ is holomorphic in this region.

In the next two sections we establish some properties of $f_m(\lambda)$ in
the upper half plane (corresponding to
$\Lambda_0$), and on its boundary.  In Section \ref{s:boundary} we show that if
$\arg \lambda=0$ or $\arg \lambda=\pi$, then
$|f_m(\lambda)|=O(|\lambda|^{d-1})$ when $|\lambda|\rightarrow \infty$.
   In Section \ref{s:lowerbounds}
we prove that there are constants 
$M,\;c_0>0$ so that for $\sigma>M>0$, $\log |f_m(e^{i \pi/2} \sigma)| \geq c_0 \sigma^d$.
Section \ref{s:ca} recalls a result of Govorov \cite{govorov}
for functions analytic in a half-plane and 
proves a consequence of this result 
which we shall need in the proof of the theorem.  Section 
\ref{s:proofofthm}  proves Theorem \ref{thm:main} by showing that the 
properties of $f_m$ established in Sections \ref{s:boundary} and 
\ref{s:lowerbounds} are inconsistent with having 
$\lim\sup_{r\rightarrow 0} \left( (\log r)^{-1} \log n_m(r)\right) <d$.

Although
the proof is different, both the result and some of the ideas underlying 
the proof of Theorem \ref{thm:main} are similar
to the results of \cite{chevenfxsgn}.
The paper \cite{chevenfxsgn}  shows that for scattering by fixed-sign
potentials in even dimensions
a lower bound like that of Theorem \ref{thm:main} holds.  
In both \cite{chevenfxsgn} and this paper, we study resonances on $\Lambda_m$
by studying zeros
of a function analytic on $\Lambda_0$. We use different complex-analytic 
results in the two papers
-- compare Govorov's results \cite[Theorem 3.3]{govorov} recalled here
in Theorem \ref{thm:govorov1}, to the results of 
\cite[Proposition 2.4]{chevenfxsgn}.  Additionally, the results 
we need for the behavior of $f_m(\lambda)=\det(mS(\lambda)-(m-1)I)$
on the boundary of $\Lambda_0$, proved here in Section \ref{s:boundary},
 are much more delicate in the obstacle
case than the corresponding results used in \cite{chevenfxsgn}.

We note that the paper
\cite{ch-hi2} proved, 
for a Schr\"odinger operator with a ``generic'' potential
$V\in L^\infty_0(\Real^d)$, lower bounds on the $m$th resonance counting function
of the type we prove here.

Section \ref{s:smat0} gives, for completeness, some basic results about
the behavior of $S(\lambda)$ near $|\lambda|=0$.

\vspace{2mm}

{\bf Acknowledgments.}  The author 
gratefully acknowledges the partial support of the NSF under grant 
DMS 1001156.  Thank you to Fritz Gesztesy for helpful conversations.

\section{A bound on $|\det(mS(\lambda)-(m-1)I)|$ for $\arg \lambda =0$
or $\arg \lambda=\pi$} \label{s:boundary}

This section uses  some results
of \cite{e-p1, e-p2, le-pe} on the 
one-sided accumulation of the eigenvalues
of the scattering matrix $S(\lambda)$ when $\arg \lambda =0$ and a 
sort of ``inside-outside duality''.   Here  ``inside-outside duality'' refers to
 a relation between the spectrum of the scattering
matrix for the exterior problem (on $\Real^d \setminus \overline{\mco}$)
and the spectrum of the interior operator, that is, the Laplacian with
corresponding boundary conditions on $\mco$.  To 
make then summary of the results which we shall need more readable, we
present them as two separate theorems, one for the Dirichlet boundary
condition and one for the Robin boundary condition. 

In this section we identify $\Lambda_0$ with the open upper half plane, and
similarly identify boundary points.  Hence $\lambda \in \Real_+$ corresponds
to a point in $\Lambda$ with argument $0$.  Recall that $S(\lambda)$ is
a unitary operator for $\lambda>0$.
\begin{thm}\label{thm:dirichlet}\cite{e-p1,e-p2,le-pe}
 Let $S(\lambda)$ denote the scattering matrix for $-\Delta$ on $\Real^d \setminus \overline{\mco}$ with 
Dirichlet boundary conditions.  Let $\lambda, \; \lambda_0 \in \Real_+$, 
$\epsilon>0$.
  Then $S(\lambda) $ has only finitely many eigenvalues 
with positive imaginary part.  Moreover, $S(\lambda)$ has
an eigenvalue $E(\lambda)$ depending continuously on $\lambda \in 
(\lambda_0-\epsilon, \lambda_0) $ with
 $\lim_{\lambda \uparrow \lambda_0}E(\lambda)=1$ 
and $\Im E(\lambda)>0$ for all
$ \lambda \in (\lambda_0-\epsilon, \lambda_0) $ if and only if $\lambda_0^2$ 
is an eigenvalue of $-\Delta$ on $\mco $ with Dirichlet boundary conditions.
Moreover, there is no pair $\lambda_0, \;E(\lambda)$ of $\lambda_0>0$ and 
 eigenvalue $E(\lambda)$ of $S(\lambda)$ 
depending continuously on $\lambda \in 
(\lambda_0,  \lambda_0+\epsilon )$
 so that $\Im_{\lambda \downarrow \lambda_0}E(\lambda)=1$ and
 $\Im E(\lambda)>0$ for all  $\lambda \in (\lambda_0,  \lambda_0+\epsilon )$.
\end{thm}

The results for the Robin-type boundary condition are similar, but the 
direction of accumulation of the eigenvalues and of the limits is different.
In the statement of the theorem, one should understand that if the boundary 
condition for the exterior problem is 
$$h(x) u (x)+\frac{\partial}{\partial n}u(x) =0 \; \text{on}\;\partial (\Real^d \setminus 
\overline{\mco})$$
with $n$ the outward pointing unit normal to
 $\partial(\Real^d \setminus \overline{\mco})$,
then the boundary condition for the interior problem is 
$$h(x) u (x)+\frac{\partial}{\partial n}v(x) =0 \; \text{on}\;\partial (
\overline{\mco})$$
 where $n$ remains the outward pointing unit normal to
 $\partial(\Real^d \setminus \overline{\mco})$.
\begin{thm}\cite{e-p2,le-pe}\label{thm:neumann}
Let $S(\lambda)$ denote the scattering matrix for $-\Delta$ on $\Real^d \setminus \overline{\mco}$ with 
Robin-type boundary conditions for 
an admissable function $h\in C^1(\partial \mco)$, $h\geq 0$. 
 Let $\lambda, \lambda_0 \in \Real_+$.
  Then $S(\lambda) $ has only finitely many eigenvalues 
with negative imaginary part.  Moreover, $S(\lambda)$ has
an eigenvalue $E(\lambda)$ depending continuously on $\lambda \in 
(\lambda_0, \lambda_0+ \epsilon) $ with 
$\lim_{\lambda \downarrow \lambda_0}E(\lambda)=1$ and $\Im E(\lambda)<0$ for all $
 \lambda \in ( \lambda_0,\lambda_0+\epsilon,) $ if and only if $\lambda_0^2$ 
is an eigenvalue of $-\Delta$ on $\mco $ with Robin-type boundary conditions.
Moreover, there is no pair $\lambda_0, \;E(\lambda)$ of $\lambda_0>0$ and 
 eigenvalue $E(\lambda)$ of $S(\lambda)$ 
depending continuously on $\lambda \in 
( \lambda_0- \epsilon, \lambda_0)$
 so that $\lim_{\lambda \uparrow \lambda_0}E(\lambda)=1$ and
 $\Im E(\lambda)<0$ for $\lambda \in ( \lambda_0- \epsilon, \lambda_0).$
\end{thm}  Results on the one-sided accumulation of eigenvalues of the 
scattering matrix can be found in \cite{e-p1, e-p2,le-pe},
with related results in, for example,
 \cite{yacs}.  The ``interior-exterior duality'' 
part of
Theorem \ref{thm:dirichlet} was proved in dimension $2$ in 
\cite{e-p1}, and then Theorem \ref{thm:neumann} was
proved for the Neumann boundary condition, again
in dimension $d=2$, in \cite{e-p2}.  The paper \cite{le-pe} proves
both Theorems \ref{thm:dirichlet} and \ref{thm:neumann} in dimension
$d=3$.  However, the proof of \cite{le-pe} works in general dimension
$d\geq 2$ with straightforward modifications.

This next proposition is central to the proof of Theorem \ref{thm:main}.
\begin{prop}\label{p:summod2pi} Let $S$ be the scattering 
matrix for the operator $P$
(with either the Dirichlet or admissable Robin boundary condition in
the exterior of an obstacle $\mco$)
and let $r>0$.  Let $\{ e^{i\theta_\alpha(r)}\}_{\alpha \in \mca}=
 \{ e^{i\theta_\alpha(r)}\}_{\alpha \in \mca(r)}$ be 
the eigenvalues of $S(r)$, repeated according to their  multiplicity.  Then
$$\sum_{\alpha \in \mca} \left( \inf_{k\in \Integers} \{ |\theta_\alpha(r)-2\pi k| \}  \right)
=O(r^{d-1})\; \text{as }\; r\rightarrow \infty.$$
\end{prop}
Before proving the proposition, we make a comment about the choice
of notation $\{ e^{i\theta_\alpha(r)}\}_{\alpha \in \mca}$. 
The eigenvalues of $S(r)$ are a countable set.
 However, in our proof it will be convenient to choose
the phases $\theta_\alpha(r)$ to be continuous
functions of $r$  when possible.  This is
possible when $e^{i\theta_\alpha(r) }$ is away from $1$.  But it can 
happen that there is an $r_0>0 $ and
an  eigenvalue $E(r)$ of $S(r)$, chosen continuous on $(r_0-\epsilon, r_0)$,
so that $\lim_{r\uparrow r_0}E(r)=1$ but $1$ is {\em not} an eigenvalue
of $S(r_0)$.  See \cite[Section 2]{e-p1}
 for examples and further discussion.  Hence 
an eigenvalue of the scattering matrix can ``disappear.''  With the 
notation $\{e^{i\theta_{\alpha}(r)}\}_{\alpha \in \mca(r)}$ we wish
to indicate the possibility of using different indexing sets for different
values of $r$.  
\begin{proof}
There is a great deal of flexibility in choosing the set of ``phases''
$\{ \theta_\alpha(\lambda)\}$; note that making a different choice (for 
example, adding an integral multiple of $2\pi$ to one or
more of the phases) does not change the value
of the  
sum in the statement of the proposition.  To prove the proposition, we 
shall make a convenient choice of this set.

For $\lambda >0$, let $\{ e^{i\theta_\alpha(\lambda)}\}$ be the eigenvalues of
$S(\lambda)$.  It is possible to 
make a choice of the set $\{\theta_\alpha(\lambda)\}$ so
that each $\theta_\alpha$ is a continuous function of $\lambda$, except, 
perhaps, where $e^{i\theta_\alpha(\lambda)}$ 
is $1$ or approaches $1$.  In the proof we choose  
 each ``phase'' $\theta_\alpha(\lambda)$
be continuous as a function of $\lambda$ except, possibly, at 
points where a one-sided limit of $e^{i\theta_\alpha(\lambda)}$ is $1$.
Moreover, we choose each phase to be defined on a maximal open interval
in $(0,\infty)$, 
so that if $\theta_\alpha$ is defined on $(\lambda_0, \lambda_1)$,
$0\leq \lambda_0<\lambda _1\leq \infty$, if $\lambda_0>0$  then
$\lim_{\lambda\downarrow \lambda_0}e^{i\theta_\alpha (\lambda)}=1.$  Similarly, if
$\lambda_1<\infty$, then $\lim_{\lambda\uparrow \lambda_1}e^{i\theta_\alpha
(\lambda)}=1.$

Additionally, we require that $\theta_\alpha(\lambda)
\in (-2\pi, 2\pi)$.  We may impose another condition on the set
$\{ \theta_\alpha\}$.  We require that if $\theta_{\alpha_0}(\lambda)$ is 
continuous on $(\lambda_0, \lambda_1), $ $0\leq \lambda_0<\lambda_1$,
 and $e^{i\theta_{\alpha_0} (\lambda)}$
is not $1$ on that same interval, but $\lim_{\lambda\downarrow \lambda_0}
e^{i\theta_{\alpha_0}(\lambda)}=1$, then 
$\lim_{\lambda\downarrow \lambda_0} \theta_{\alpha_0}(\lambda)=0$.  We note that
in particular this means that that if $\theta_\alpha$ is defined
on $(0,\lambda_0)$, some $\lambda_0>0$, then 
$\lim_{\lambda \downarrow 0} \theta_{\alpha}(\lambda)=0$; see Corollary 
\ref{cor:lat0}.

Having chosen these conventions, 
\begin{align} \label{eq:scph}
\frac{1}{2\pi i} \int_0^r \frac{d}{d\lambda} \log \det S(\lambda) d\lambda 
  = \frac{1}{2\pi} \sum \theta_\alpha(r) + N(2\pi, r) - N(-2\pi, r) +O(1)
\end{align}
as $r\rightarrow \infty$.  Here we use the notation
\begin{multline*}
N(\pm 2\pi, r) \\
=\#\{ \lambda_0,\; 0<\lambda_0 \leq r\; : \lim_{\lambda 
\uparrow \lambda_0}\theta_\alpha(\lambda) = \pm 2\pi\; 
\text{for some $\alpha$, counted with multiplicity }\}.
\end{multline*}

Now we specialize to the case of admissable Robin-type boundary conditions.
By results of \cite{e-p2, le-pe} recalled here in Theorem \ref{thm:neumann},
\begin{align*} & 
N(-2\pi, r)+ \# \{ \alpha \in \mca(r): \theta_\alpha(r)<0\}\\
& = \# \{ \lambda: 0<\lambda \leq r\; \text{and}\; 
\lambda^2 \; \text{is a Robin eigenvalue of $-\Delta$ on $\mco$ } \}+O(1).
\end{align*}
By the well-known Weyl formula for the Laplacian on a bounded open set with
smooth boundary,
this means
\begin{equation}\label{eq:negative}
N(-2\pi, r) + \# \{ \alpha \in \mca(r): \theta_\alpha(r)<0\} =
 c_d \vol(\mco) r^d+O(r^{d-1})
\end{equation}
where $c_d$ is the $d$-dimensional Weyl constant.
Then
$$ \frac{1}{2\pi i} \int_0^r \frac{d}{d\lambda} \log \det S(\lambda) d\lambda
\geq \frac{1}{2\pi} \sum_{\theta_\alpha (r)<0} \theta_\alpha(r) - N(-2\pi, r) +O(1)
\geq - c_d \vol(\mco ) r^{d}+O(r^{d-1}).$$
On the other hand \cite{buslaev, compsupp, m-r, melrose, robert}
\begin{equation}\label{eq:weyl}
\frac{1}{2\pi i} \int_0^r \frac{d}{d\lambda} \log \det S(\lambda) d\lambda= 
-c_d \vol(\mco) r^d+ O(r^{d-1} ) \; \text{as} \; r\rightarrow \infty .
\end{equation}
Thus we must have 
 \begin{equation}\label{eq:negph}
\frac{1}{2\pi}\sum_{\theta_\alpha (r)<0} \theta_\alpha(r) - 
N(-2\pi, r)= -c_d \vol(\mco) r^d+ O(r^{d-1} )
\end{equation}
and this, together with (\ref{eq:negative}), means that
$$\sum_{\theta_\alpha(r)<0} (1+\frac{1}{2\pi} \theta_\alpha(r))= O(r^{d-1})$$
or
\begin{equation}\label{eq:negest}
\sum_{\theta_\alpha(r)<0} |\theta_\alpha(r)+2\pi| = O(r^{d-1}).
\end{equation}

Using (\ref{eq:scph}) and (\ref{eq:weyl}- \ref{eq:negest}), we obtain that
$$\frac{1}{2\pi}\sum_{\theta_\alpha(r)>0} \theta_\alpha(r)+ N(2\pi, r) = 
O(r^{d-1}).$$
This finishes the proof of the proposition for the Robin case.

The proof for the Dirichlet case is similar, using Theorem \ref{thm:dirichlet}.
\end{proof}

\begin{prop}
Let $m\in \Natural$, and $r>0$.  
Then there is a constant $C>0$ so that 
$$1\leq |\det(mS(r)-(m-1)I)|\leq \exp(C r^{d-1})\; \text{for $r$ sufficiently
large}$$
and 
$$|\det(mS(re^{i\pi})-(m-1)I)|\leq \exp(C r^{d-1})\; \text{for $r$ sufficiently
large}.$$
\end{prop}
We note that the constant $C$ depends on $m$ as well as on $\mco$ and the 
boundary condition.
\begin{proof}
We denote the set of eigenvalues of the scattering matrix, repeated according
to their multiplicity, by 
$S(r)$ by $\{ e^{i \theta_j(r)}\}_{j\in \Natural}$.  Since this time 
the we do not require continuity properties of $\theta_j$ we use the
index set $\Natural$.   Then
\begin{equation}\label{eq:det1}
\det(mS(r)-(m-1)I)= \prod ( 1+ m(e^{i\theta_j(r)}-1)).
\end{equation}

Again, we make use of the fexibility in choosing the set $\{ \theta_j(r)\}$.
Here we do {\em not} need continuity properties, so we can assume, without 
loss of generality, that
$$-\pi \leq \theta_j(r) < \pi.$$
We shall split the product in (\ref{eq:det1}) into two pieces, depending
on the size of $|\theta_j(r)|$.  From Proposition \ref{p:summod2pi}, the
number of $j$ so that $|\theta_j(r)|> \epsilon>0$ is 
$O_\epsilon(r^{d-1})$.  
Since $1\leq | 1+m(e^{i\theta_j}-1)|\leq 1+2m,$ 
$$
0\leq \log\left|\prod_{ |\theta_j(r)| >1/8m}  ( 1+ m(e^{i\theta_j(r)}-1))\right|
= \sum_{|\theta_j(r)| >1/8m} \log | ( 1+ m(e^{i\theta_j(r)}-1))| 
= O(r^{d-1}).$$
Now 
\begin{align*}
\prod_{|\theta _j(r)|\leq 1/8m} ( 1+ m(e^{i\theta_j(r)}-1))
& = \exp\left( \sum _{ |\theta_j(r)| \leq 1/8m} 
\log ( 1+ m(e^{i\theta_j(r)}-1) )\right)
\\ & = \exp\left( \sum_{|\theta_j(r)|\leq 1/8m}(i m \theta_j(r) +O(|\theta_j(r)|^2)
\right)
\end{align*}
Of course, $\left| \exp \left( \sum_{|\theta_j(r)|<1/8m} i m \theta_j(r) \right) 
\right| =1$.  Moreover, since if $|\theta_j(r)|< 1/8m$, then
$|\theta_j(r)|^2<| \theta_j(r)|$, 
and by Proposition 
 \ref{p:summod2pi}
$$\sum _{|\theta_j(r)|\leq 1/8m} \left|   \theta_j(r)\right| ^2
\leq \sum_{|\theta_j(r)|\leq 1/8m}   |\theta_j(r)|  = O(r^{d-1})
$$
we have that the term $\sum_{|\theta_j(r)|\leq 1/8m} O(|\theta_j(r)|^2)
= O(r^{d-1}).$
This finishes the proof of the first statement.

To prove the second inequality in the proposition, note that by 
\cite[Proposition 2.1]{ch-hi4}, $S(re^{i\pi})=2I -\mcr S^*(r)\mcr$, 
where $\mcr:L^2(\Sphere^{d-1})\rightarrow L^2(\Sphere^{d-1})$ is 
$(\mcr f)(\theta)=f(-\theta)$, and $S^*$ is the adjoint of $S$.  Hence
\begin{align*}
\det(mS(re^{i\pi})-(m-1)I) & = \det (m (2I-\mcr S^*(r) \mcr)-(m-1)I)\\
 & = \det((m+1)I -m\mcr S^*(r)\mcr)\\
 & = \det ( (m+1)I-mS^*(r)).
\end{align*}
Hence, if as before we denote the eigenvalues of $S(r)$ by
 $\{ e^{i\theta_j(r)}\}$, then 
$$\det(mS(re^{i\pi})-(m-1)I)= \prod \left( 1-m(e^{-i\theta_j(r)}-1)\right)$$
and the proof of the second inequality follows essentially the same way as 
the proof of the first one.
\end{proof}

\section{Lower bounds on $|\det(mS(\sigma e^{i\pi/2})-(m-1)I)|$ when $\sigma
\rightarrow \infty$, $\sigma \in (0,\infty)$} 
\label{s:lowerbounds}
The main result of this section is Proposition \ref{prop:lowerbd}, which
provides a lower bound on $|\det (mS(e^{i\pi/2}\sigma)-(m-1)I)|$ when $\sigma>0$,
$\sigma \rightarrow \infty$.  The 
proof of this proposition uses three main ideas: 
the fact that $S(i\sigma)-I$ has purely imaginary eigenvalues when $d$ is even
and $\sigma$ is sufficiently large;
a monotonicity-type result of \cite{beale,la-ph}; and explicit calculations
in the case of a ball along with properties of Bessel functions.

In this section we work in $\Lambda_0$, which we identify with the open upper
half plane of $\Complex$.  Hence for $\sigma>0$, $i\sigma$ corresponds
to $e^{i\pi/2}\sigma$.

In this section we shall make use of some results of \cite{beale, la-ph}.  
We note that as our choice of the physical half plane is different from
theirs (we choose $0<\arg \lambda<\pi$ as the physical region,
they choose $-\pi<\arg \lambda <0$) some notation will be a bit different.

We recall some basic definitions related to 
the scattering matrix. 
For $\lambda \in \Complex $ with $0\leq \arg \lambda\leq \pi$ and $\omega
\in \Sphere^{d-1}$,
there is a unique solution to the equation
$$(-\Delta -\lambda^2) v=0\; \text{in}\;  \Real^d \setminus \mco$$
satisfying either the boundary condition (Dirichlet type)
$$v\restrict_{ \partial \mco} = e^{-i\lambda x\cdot \omega }\restrict_{
\partial \mco}$$
{\em or }  satisfying the Robin-type boundary condition
$$h(x) v(x)+\frac{\partial v(x)}{\partial n}= h(x) e^{-i\lambda x\cdot \omega }
+
\frac{\partial }{\partial n}e^{-i\lambda x\cdot \omega } \; \text{on}\; \partial \mco.$$  Here $h\in C^1(\partial \mco)$, $h\geq 0$, and  $n$ is the outward unit
normal to $\Real^d \setminus \overline{\mco}$.  
In addition, to guarantee uniqueness, we require that $v$ satisfy
a radiation condition at infinity: if 
$\mco \subset B(0;R)= \{ x\in \Real^d : |x|<R\}$, then 
$$\left(  \frac{\partial}{\partial |x|} v -i\lambda v \right)
 \restrict_{
\Real^d \setminus \overline{B}(0;R)}\in L^2(\Real^d \setminus \overline{B}(0;R))
.$$
It follows then that for large $|x|$, $v$ has the form
$$v(x; \omega, \lambda)= |x|^{-(d-1)/2}
e^{i\lambda |x|}\left( k(\omega, x/|x|, \lambda)+O(|x|^{-1})\right).$$
This function $k$ is called the {\em transmission coefficient}.

Now the scattering matrix $S(\lambda)$ is
given by $S(\lambda)=I+K(\lambda)$, where 
\begin{equation}\label{eq:K}
[K(\lambda) f](\omega)=
-\left( \frac{i\lambda}{2\pi}\right)^{(d-1)/2} 
\int_{\Sphere^{d-1}} k(\omega, -\theta;\lambda)f(\theta)d\theta.
\end{equation}

The proof of the following lemma uses separation of 
variables and explicit computations involving Bessel and Hankel functions.
Related calculations have been made in many places, including
\cite{beale,la-ph,chevenfxsgn,vasy}.
\begin{lemma}\label{l:spherecalc}
Let $\mco= B(0;R)$, and let $S(\lambda)$ denote 
the scattering matrix for $-\Delta$ on $\Real^d\setminus
\overline{\mco}$ with either Dirichlet or Neumann 
boundary conditions. Let $d$ be even 
and $m\in \Natural$.  Then there is a constant $c_0>0$, depending
on $R$ and $m$, so that
for $\sigma >0$, $|\det (mS(i\sigma)-(m-1)I)|\geq c_0 \exp(c_0 \sigma ^d).$
\end{lemma}
\begin{proof}
Let $\{ Y^\mu_l \}$, $l=0,1,2,...$, $\mu=1,\;2,\;..., \mu(l)$ be a complete 
orthonormal
set of
spherical harmonics on $\Sphere^{d-1}$.
 Here $\mu(l)= \frac{2l+d-2}{d-2}\binom{l+d-3}{d-3}$  
and these eigenfunctions of the Laplacian $\Delta_{\Sphere^{d-1}}$
on $\Sphere^{d-1}$ satisfy
$$-\Delta _{\Sphere^{d-1}}Y^\mu  _l=l(l+d-2)Y^\mu_l,\;
 l=0,\; 1,\;2,...,\; \mu=1,\; 2,\;...,\; \mu(l).$$
For the Dirichlet Laplacian on $\Real^d \setminus \overline{B}(0;R)$,
the transmission coefficient $k_D(\lambda)$ is
$$k_D(\theta,\theta',\lambda)= 2\left( \frac{2\pi}{\lambda}\right)^{(d-1)/2}
\sum_{l=0}^{\infty} \sum _{\mu =1}^{\mu(l)} (-i)^l e^{-i(\nu \pi/2+ \pi/4)}
   \frac{J_\nu(\lambda R)}{H^{(1)}_\nu (\lambda R)} Y^\mu_l(\theta)\overline{Y}^\mu_l(\theta') $$
where $\nu =\nu(l)= l-1+d/2$, and $J_\nu$ is the Bessel function of order
$\nu$ of the first kind, and $H^{(1)}_\nu$ is a Hankel function.
We are interested in $k_D$ evaluated at $\lambda=i\sigma$.  Using 
\cite[9.6.3 and 9.6.4]{olver}
$$k_D(\theta,\theta',i\sigma)= \pi
\left( \frac{2\pi}{\sigma}\right)^{(d-1)/2}
\sum_{l=0}^{\infty} \sum _{\mu =1}^{\mu(l)} 
   \frac{I_\nu(\sigma R)}{K_\nu (\sigma R)} 
Y^\mu_l(\theta)\overline{Y}^\mu_l(\theta').$$
Now we note that since the eigenvalues of $k_D(i\sigma)$ are 
real, and the spherical harmonics are either even or odd in the reflection
$\omega \rightarrow -\omega$, the eigenvalues of $K_D(i\sigma)$ 
are pure imaginary, and the eigenvalues
of $(2\pi/\sigma)^{(d-1)/2}K_D(i\sigma)$ have the same norm as the eigenvalues
of $k_D(i\sigma)$.  
Hence
\begin{align*}
|\det (mS(i\sigma)-(m-1)I)|& 
= \prod_{l} \left( 1+ \left| \pi m \frac{I_\nu(\sigma R)}{K_\nu (\sigma R)}
\right|^2\right)^{\mu(l)/2}.\\
 & = 
 \exp \left[\sum_{l=0}^\infty
 \frac{\mu(l)}{2} \log   \left( 1+ \left| \pi m \frac{I_\nu(\sigma R)}
{K_\nu (\sigma R)}
\right|^2\right) \right] \\ & \geq 
\exp \left[ \sum_{\sigma R/M \leq l \leq \sigma R}
\frac{\mu(l)}{2} \log   \left( 1+ \left| \pi m \frac{I_\nu(\sigma R)}{K_\nu (\sigma R)}
\right|^2\right) \right]
\end{align*}
for $M>1$.
From the uniform asymptotic expansions of \cite[9.7.7,9.7.8]{olver},
we have that for $\tau>0$ in a fixed compact set
$$\frac{I_\nu(\nu \tau)}{K_\nu(\nu \tau)} = \frac{1}{\pi} e^{2 \nu \eta}
( 1+O(\nu^{-1}))$$
where $\eta = \sqrt{1+\tau^2}+ \log ( \tau/(1+\sqrt{1+\tau^2}))$.
By restricting $\sigma R/M \leq l \leq \sigma R$ for
some finite  $M>1$, thus ensuring $\tau= \sigma R/\nu$ lies in a compact
set away from $0$, we get from these
asymptotics that for sufficiently large $\sigma$
\begin{align*}|\det (mS(i\sigma)-(m-1)I)| & \geq 
  \exp \left( \sum_{\sigma R/M \leq l \leq \sigma R} 
\frac{\mu(l)}{2} l \right)  \\ 
& \geq c_0 \sigma^{d}.
\end{align*}
 The last inequality uses that 
$\mu(l) > c_0' l^{d-2}>0$ for sufficiently large $l$.

Likewise, for the Robin-type boundary conditions in the exterior of 
the sphere where the boundary function is $h_0/R$, for a 
constant $h_0\geq 0$
$$k_{h_0}(\theta,\theta',i\sigma)= \pi \left( \frac{2\pi}{\sigma}\right)^{(d-1)/2}
\sum_{l=0}^{\infty} \sum _{\mu =1}^{\mu(l)} 
\frac{ (h_0+\frac{d-2}{2})I_\nu(\sigma R)-\sigma R I_\nu ' (\sigma R)}{
(h_0+\frac{d-2}{2}) K_\nu(\sigma R)+\sigma R K'_\nu (\sigma R)}
Y^\mu_l(\theta)\overline{Y}^\mu_l(\theta').
$$
Thus, a similar computation as in the Dirichlet case, using 
\cite[9.7.7-9.9.10]{olver}, gives the result for the Neumann boundary
condition ($h_0=0$), or indeed for any Robin-type boundary condition
with $h_0 \geq 0$.
\end{proof}

We recall some results of \cite{beale, la-ph} which we shall use.
The first is \cite[Theorem 3.7]{beale} along with 
some results of \cite[Theorem 3.5]{beale}, which generalizes 
\cite[Theorem 2.4]{la-ph}.  In
the statement of the theorem we use $k(i\sigma)$ to denote the 
operator given by $[k(i\sigma)f](\omega)= \int_{\Sphere^{d-1}}k(\omega,\theta;
i\sigma) f(\theta)d\theta$.  This operator $k(i\sigma)$ is self-adjoint
for the boundary conditions which we consider.

\begin{thm}\cite[Theorem 3.7; see also Theorem 3.5]{beale}  
\label{thm:bealemono}
Let $\mco_1$ and $\mco_2$ be obstacles so that $\overline{\mco_1}
\subset \mco_2$.  Let $h_j$ be admissable boundary functions 
on $\partial \mco_j$, $j=1,2$, and let $k_j(\lambda)$ denote 
 the operators on $L^2(\Sphere^{d-1})$  with Schwartz 
kernels given by the 
 corresponding transmission coefficients for
the Robin boundary conditions.  
Then there is a $\sigma_0>0$
depending on $\mco_1$, $\mco_2$, and on $h_2$ so that $0>k_1(i \sigma) 
\geq k_2(i\sigma)$ for $\sigma>\sigma_0$.  For obstacles $\mco_1$, 
$\mco_2$, $\overline{\mco_1}\subset \mco_2$, with Dirichlet
boundary conditions, $0<k_1(i\sigma)\leq k_2(i\sigma)$ for 
all $\sigma>0$.
\end{thm}

Let $\mcr:L^2(\Sphere^{d-1})\rightarrow L^2(\Sphere^{d-1})$ be defined
by $(\mcr f)(\theta)= f(-\theta)$.  The following proposition follows
immediately from Theorem \ref{thm:bealemono} and 
 \cite[Theorem 4.4]{la-ph}.
\begin{prop} \label{p:l-pmono}  Let $\mco_j$,
 $h_j$, and $k_j$ be as in the statement of Theorem \ref{thm:bealemono}.
Then the eigenenvalues of $(k_j\mcr)(i\sigma)$ are real for 
sufficiently large $\sigma>0$.  Order the eigenvalues
of $(k_j  \mcr) (i\sigma) $, taking account of multiplicities:
$$\nu_1^{(j)}(i\sigma) \geq \nu_2^{(j)}(i\sigma) \geq ... >0 > ...\geq \kappa_2^{(j)}(i\sigma)\geq 
\kappa_1^{(j)}(i\sigma),\; j=1,\; 2$$
Then there is a $\sigma_0\geq 0$ so that for $\sigma>\sigma_0$
and for each $n\in \Natural$, 
$$ \nu_n^{(1)}(i\sigma) \leq \nu_n^{(2)}(i\sigma)\; \text{and}\; 
\kappa_n^{(1)}(i\sigma) \geq \kappa_n^{(2)}(i\sigma).$$
\end{prop}

We are now ready to prove the main result of this section.
\begin{prop}\label{prop:lowerbd} Let $d$ be 
even and let $\mco \subset \Real^d$ be an obstacle with
$\mco \not = \emptyset$.  
Let $S(\lambda)$ denote the scattering matrix of $-\Delta$ on $\Real^d \setminus
\mco$ with either Dirichlet or Robin type boundary conditions.  In the 
latter case assume the boundary function $h$ satisfies
 $h\in C^1(\partial \mco)$, $h\geq 0$.
Then for $m \in \Natural $
there is a constant $c>0$, depending on both $\mco $ and 
the boundary condition, so that for sufficiently large $\sigma >0$,
$|\det (m S(i\sigma)-(m-1)I)| \geq \exp (c |\sigma|^d)$.
\end{prop}
\begin{proof}
Our conditions on $\mco$ ensure that there is some nontrivial 
closed ball contained in $\mco$.  By translating if necessary, we may assume the ball is $\overline{B}(0;R)$ for some $R>0$.  

We shall apply Proposition \ref{p:l-pmono} with $\mco_1 =B(0;R)$ and
$\mco_2=\mco$.  We use the original boundary condition on 
$\Real^d \setminus \overline{\mco_2}=\Real^d \setminus \overline{\mco}$.   
If the original boundary condition is Dirichlet, we use the 
Dirichlet boundary condition on $\Real^d \setminus \overline{\mco_1}$; if the 
original boundary condition is Neumann or Robin type, we use the Neumann boundary
condition on $\Real^d \setminus \overline{\mco_1}$.  In each case we denote the corresponding
scattering matrices by $S_j$ and the transmission coefficient
and its corresponding operator 
on $L^2(\Sphere^{d-1})$ by $k_j$.

The eigenvalues of $S_j(i\sigma)-I$ are, by (\ref{eq:K}), given by $-i^{d-1} \left( \frac{\sigma}{2\pi}\right)^{(d-1)/2}$ times the eigenvalues of $k_j(i\sigma)\mcr$.  If we denote the 
eigenvalues of $k_j(i\sigma)\mcr$ by $\{ \nu_n^{(j)}(i\sigma)\}
\cup \{ \kappa_j^{(j)}(i \sigma)\}$ 
with the same ordering as in the statement of Proposition \ref{p:l-pmono}, we 
have, for $\sigma>0$ sufficiently large,
\begin{align*}
& |\det(mS(i\sigma)-(m-1) I)\\ & =  |\det(mS_2(i\sigma)-(m-1) I)|\\
& = \prod_{n=1}^\infty
\left| 1-m i^{d-1} \left( \frac{\sigma}{2\pi}\right)^{(d-1)/2}
\nu_n^{(2)}(i\sigma)\right| \; \prod_{j=1}^\infty
\left| 1-m i^{d-1} \left( \frac{\sigma}{2\pi}\right)^{(d-1)/2}
\kappa_j^{(2)}(i\sigma)\right| 
\end{align*}
Now from Proposition \ref{p:l-pmono}, for $\sigma>0$ large enough,
\begin{align*} & |\det(mS_2(i\sigma)-(m-1) I)| \\ &  \geq 
 \prod_{n=1}^\infty \left| 1-m i^{d-1} \left( \frac{\sigma}{2\pi}\right)^{(d-1)/2}
\nu_n^{(1)}(i\sigma)\right| \; \prod_{j=1}^\infty \left| 1-m i^{d-1} \left( \frac{\sigma}{2\pi}\right)^{(d-1)/2}
\kappa_j^{(1)}(i\sigma)\right|\\ & = |\det (mS_1(i\sigma)-(m-1)I)|.
\end{align*}
Since $\mco_1=B(0;R)$, Lemma \ref{l:spherecalc} finishes the proof.
\end{proof}

\section{Complex-analytic results} \label{s:ca}

In this section we denote by $U_+$ the upper half plane:
$U_+=\{ z\in \Complex: \Im z>0\}$.  Let $f:U_+\rightarrow \Complex
$ be an analytic function, not identically $0$.  Assume
in addition that $f$ is continuous on $\overline{U_+}\setminus\{0\}$
and bounded in compact sets of $U_+$.
   Then we define (e.g. \cite[page 5]{govorov} or \cite[Section 1.14]{levin})
the 
{\em order} of $f$ in $U_+$ to be 
$$\rho = \lim \sup _{r\rightarrow \infty} 
\frac{\log_+ \log _+ \sup_{z\in U_+,\; |z|\leq r}|f(z)|}{\log r}.$$
We shall be interested in the case where $\rho>1$ is finite. When 
$\infty>\rho>1$ 
 this definition of order is consistent with
the definition of order in an angle given by Govorov in \cite[Part I, 
Section 1, page 1]{govorov};
see \cite[Theorem 1.4]{govorov}.

In preparation for
 the next theorem, we introduce some notation.  For $q\in \Natural$,
let $$E_q(u)= (1-u) \exp \left( u+\frac{u^2}{2}+... + \frac{u^q}{q}
\right) $$
denote the canonical Weierstrass factor.  Set, for $a\in \Complex$, $a\not =0$,
$$D_q(u,a)= \frac{E_q(u/a)}{E_q(u/\overline{a})},$$
the canonical Nevanlinna factor.

We shall use an adaptation of a result of \cite{govorov} for a function $f$
analytic on $U_+$,
of finite order $\rho>1$, which, in addition, has a continuous extension to
$\overline{U_+}$.  We note that this last condition (the continuous extension
to $\overline{U_+}$) is not made in \cite[Theorem 3.3]{govorov}, but allows us
to simplify the statement of the result-- in particular, with this 
condition, the singular boundary function, denoted
by $\varphi$ in the statement of \cite[Theorem 3.3]{govorov}, is
identically constant.  

\begin{thm}(\cite[Theorem 3.3, adapted; see also Theorem 3.2]{govorov}
\label{thm:govorov1}
 Let $f$ be 
analytic in the half plane $U_+$ and continuous on $\overline{U_+}.$
Suppose $f$ is of finite order $\rho$, $\rho \geq 1$, in $U_+$.
Let $\{z_n\}= \{ r_ne^{i\theta_n}\}$, $0<\theta_n<\pi$ be the set of zeros of 
$f$ in $U_+$, and set $q=[\rho]$.  Then there are real constants
$a_0,\; a_1,\; ...,a_q$ so that 
\begin{multline} f(z) =
\exp\left( i(a_0+a_1z+...+ a_qz^q) +\frac{1}{\pi i} \int_{-1}^1 \frac{ \log |f(t)|}{t-z}dt\right) 
\prod _{|z_n|\leq 1} \frac{z-z_n}{z-\overline{z_n}} \\ \times
\prod _{|z_n|>1}  D_{q}(z,z_n) 
\times \exp\left( \frac{z^{q+1}}{\pi i}\int_{|t|\geq 1} \frac{\log |f(t)|}
{t^{q+1}(t-z)}dt \right).
\end{multline}
The integrals and products in this expression are absolutely convergent.
Moreover,
$$ \sum_{r_n\leq 1} r_n \sin \theta_n <\infty,\; \sum_{r_n>1} r_n ^{-\rho -\epsilon}
\sin \theta_n<\infty, \; \int_{-\infty}^\infty
\frac{|\log |f(t)||}{1+|t|^{1+\rho+\epsilon}}dt<\infty$$
for any $\epsilon >0$.
\end{thm}
We remark for those comparing \cite[Theorem 3.3]{govorov} there seems to be a small error-- there does not seem to 
be a reason that the constants $a_j$ cannot be negative.
We include the restriction $\rho>1$ here because it is for such 
$\rho$ that our definition of order in a half plane
coincides with that of \cite{govorov}.

\begin{prop} \label{p:upperbd} Let $f$ be a function analytic in $U_+$ and continuous on 
$\overline{U_+}$, and of order at 
most $\rho>1$ in $U_+$.  Let 
$$\tilde{n}_f(r) =
\# \{ a_j \in U_+  :\; f(a_j)=0,\; \text{counted with multiplicity}\}.$$
Suppose $\tilde{n}_f(r)=O(r^{\rho'})$ as $r\rightarrow 
\infty$  and $\log |f(\pm t)|= O(t^{\rho'})$ as $t\rightarrow \infty$ for 
some $\rho'<\rho$.  Suppose $[\rho]=q$ is  even.
Let $\sigma>0$.  Then for any $\epsilon >0$ there is a constant $C=C_{\epsilon}<\infty$ so that
$$\log |f(i\sigma)| \leq C (1+ \sigma ^{\max(\rho'+\epsilon, q-1)} ),\; \sigma >0.$$
\end{prop}
We comment that the restriction that $[\rho]=q$ be even is necessary. 
For  odd $q$,
we may consider as a counterexample the function $\exp(\pm i z^q)$, where the choice of 
sign is determined by the parity of $(q-1)/2$.
\begin{proof}
In this proof $C$ denotes a positive constant which may depend upon $\epsilon$
and may change from line to line.

We use the expression for $f$ from Theorem \ref{thm:govorov1}, along with the notation of that theorem.  In 
particular, $\{z_n\}$ denotes the set of zeros of $f$ in $U_+$, repeated according to their multiplicity.
From Theorem \ref{thm:govorov1}, we 
can write 
$$f(z) = \exp(g_1(z)+ g_2(z)) \prod _{|z_n|\leq 1}\frac{z-z_n}{z-\overline z_n}
 \prod _{|z_n|>1 } D_q(z,z_n)$$
where 
$$g_1(z) = i(a_0+a_1z+...+a_q z^q)+ \frac{1}{\pi i}\int_{-1}^1 \frac {\log |f(t)|}{
t-z} dt$$
and $$g_2(z) = \frac{z^{q+1}}{\pi i}\int_{|t|\geq 1} \frac{\log |f(t)|}
{t^{q+1}(t-z)}dt.$$

Recalling that $a_j \in \Real$ and $q$ is even, we see 
$$ \Re g_1(i\sigma) =O(\sigma^{q-1})\; \text{as} \; \sigma \rightarrow 
\infty$$  
so that $|\exp(g_1(i\sigma))|\leq C\exp(C\sigma^{q-1})$.  Moreover, 
\begin{align}
\Re (g_2(i\sigma)) & = \frac{i^q\sigma^{q+1}}{2\pi}
\int_{|t|\geq 1} \frac{ t \log|f(t)|}{t^{q+1}(t^2+\sigma^2)} dt .
\end{align}
Thus, if $\rho'<q-1$, we see immediately that $\Re g_2(i\sigma)=O(\sigma^{q-1})$
since in this case $t^{-q}\log|f(t)|$ is integrable on 
$\{ t\in \Real:\; |t|\geq 1\}$.
On the other hand, if $\rho'\geq q-1$, then for $\epsilon>0$ sufficiently
small we can write 
\begin{align*}
|\Re (g_2(i\sigma)) | 
& \leq  \frac{\sigma^{q+1}}{2\pi}
\int_{|t|\geq 1} \frac{ |\log|f(t)||}{ t^{q+1+\rho'-q +\epsilon }\sigma^{1-(\rho'-q+\epsilon )}} dt\\
 & \leq C \sigma^{\rho'+\epsilon}.
\end{align*}

Consider 
$$\prod_{|z_n|> 1} D_q(i\sigma, z_n).$$  We divide this into 
two cases, depending on the relative size of $\rho'$ and $q$.
We note that {\em if} $\rho'<q$, then 
$$\sum_{|z_n|>1}\left( \frac{ 1}{z_n^q}-  \frac{1}{\overline{z_n}^q}\right) =
\sum_{|z_n|>1}\left( \frac{ -2i \Im z_n^q }{|z_n|^{2q}}\right)
$$
and that the assumption that $q>\rho'$ implies that the sum converges.  
Hence, if $\rho'<q$, we have 
$$\prod_{|z_n|> 1} D_q(i\sigma, z_n)
= \left( \prod_{z_n > 1} D_{q-1}(i\sigma, z_n)\right) \exp \left(-i^{q+1}
 \sigma^q
\sum_{|z_n|> 1}\left( \frac{ 2 \Im z_n^q }{|z_n|^{2q}}\right)\right). $$
Since $q$ is even, the exponent in the second factor is pure imaginary.
Now when $\rho'<q$, the estimates used in the proof of 
\cite[Lemma 3.4]{govorov} (and related to fairly standard estimates of 
canonical products; compare \cite[Section I.4]{levin}, for example) 
show that for $\rho'<q$,
$$\left| \prod_{z_n > 1} D_{q-1}(i\sigma, z_n)\right| 
\leq C \exp(C|\sigma|^{\max(q-1,\rho'+\epsilon)}),\;  \epsilon >0.$$

On the other hand, if $\rho'\geq q$, a direct application
of the estimates as in the
proof of \cite[Lemma 3.4]{govorov} shows that
$$\left| \prod_{z_n > 1}D_{q}(i\sigma, z_n)\right| 
\leq C \exp(C|\sigma|^{\rho'+\epsilon}),\;  \epsilon >0.$$
In either case, we have
$$ \left| \prod_{z_n > 1}D_{q}(i\sigma, z_n)\right| 
\leq C \exp(C|\sigma|^{\max(\rho'+\epsilon,q-1)}),\;  \epsilon >0.$$

\end{proof}

\section{Proof of Theorem \ref{thm:main}}\label{s:proofofthm}


In this section we give the proof of Theorem \ref{thm:main},.
Suppose $m\in \Natural$.   We 
will apply Proposition \ref{p:upperbd} to 
$f_m(\lambda)= \det( mS(\lambda)-(m-1)I)$.

We first show that $f_m$ has the regularity properties of Proposition 
\ref{p:upperbd}.  Our assumptions on the boundary conditions on
$\partial (\Real^d \setminus \overline{\mco}) $ ensure that
the resolvent $R(\lambda)=(P-\lambda^2)^{-1}$ is holomorphic in
the closure of $\Lambda_0\subset \Lambda$.  Thus $S(\lambda)$ is 
holomorphic in that region as well. 
Morover, $S(\lambda)$ is continuous at $0\in \{ z\in 
\Complex: \Im z\geq 0\}$, see Proposition \ref{p:contat0}.  
Then $f_m$ inherits these regularity 
properties of the scattering matrix.



The proof of the theorem is by contradiction.  Suppose there
is some combination of nontrivial obstacle $\mco$ and boundary condition
(Dirichlet or admissable Robin) so that
$\lim \sup _{r\rightarrow \infty} \frac{ \log n_m(r)}{\log r} = \rho' <d$. 
Let $\tilde{n}_{f_m}(r)$ be the the number of 
 zeros of $f_m=\det(mS(\lambda)-(m-1)I)$ 
in the upper half plane with norm at most $r$.
By Proposition \ref{p:chi-hi4},
$\lim \sup_{r\rightarrow \infty}\frac{ \log n_m(r)}{\log r}
= \lim\sup _{r\rightarrow \infty}\frac{ \log \tilde{n}_{f_m}(r)}{\log r}.$
 It follows from the 
same arguments as in, for example,
 \cite[Section 2]{zw-pfeven} or \cite[Theorem 7]{zwpf}, 
that the order
of $f_m(\lambda)=\det(mS(\lambda)-(m-1)I)$ on the upper half plane is at most $d$.  Then,
using Propositions \ref{p:summod2pi} and \ref{p:upperbd}, we must have 
$|\det(mS(i\sigma)-(m-1)I)|\leq C \exp(\max (\rho'+\epsilon, d-1))$ for all $\epsilon >0$
with some constant $C=C_\epsilon$.  But this contradicts Proposition \ref{prop:lowerbd},
proving the theorem for $m>0$.

If $m<0$, we observe that $n_{-m}(r)=n_m(r)$, as follows, for example, from 
\cite[(2.2)]{ch-hi4}:
$$S(|\lambda|e^{-i\pi \arg \lambda})^*= 2I-\mathcal{R} S(e^{i\pi} \lambda) \mathcal{R},\; \lambda \in \Lambda$$
where $(\mathcal{R}f)(\theta)=f(-\theta).$

\vspace{10mm}

\section{The scattering matrix at $0$ in even dimensions}\label{s:smat0}

The results of this section, while used in the proof of Theorem \ref{thm:main},
use rather different techniques than the majority of this paper.  Hence we
include them here so as to not interupt the flow.  We note that both
Proposition \ref{p:contat0} and Corollary \ref{cor:lat0}
 may be well known, but we are unaware 
of a reference in which it is proved in this setting.  
In Section \ref{s:proofofthm} we used the fact that the scattering matrix 
has continuous extension to $\{ z\in \Complex: \Im z\geq 0\}$.  In 
this section we prove this.  With the assumptions
we have made on the operator $P$, 
the only real issue is the behavior of the scattering 
matrix at $0$.  Note that  the nature of the 
singularity of the ``model resolvent'' $(-\Delta -\lambda^2)^{-1}$
at $\lambda=0$ depends on the dimension
 and that the expression for the scattering matrix
(for example, \cite[Proposition 2.1]{pe-zwsc}, recalled
here in Proposition \ref{p:pe-zw}) has 
dimensional-dependent powers of $\lambda$.  Thus one expects the 
scattering matrix to be ``more regular'' at $0$ in higher dimensions;
compare, for example, the papers \cite{b-g-d,jensenspec,jensenspec4,j-k}
which include much more detailed results on the behavior of the 
resolvent and scattering
matrix at $0$ for the Schr\"odinger
operator.
  However, as our proofs do not depend on the dimension
other than through its parity, we give them here for all even dimensions.

The proofs we include here do not require that $S$ be the scattering 
matrix for the Laplacian with Dirichlet or (admissable) Robin-type boundary
conditions in the exterior of an obstacle.  In fact, the proofs work for
the scattering matrix for any self-adjoint operator $P$ which is 
a compactly-supported ``black-box'' perturbation of the Laplacian 
satisfying the conditions of Sj\"ostrand-Zworski, see \cite{sj-zw}.  We 
recall these assumptions for the reader's convenience.

In recalling the assumptions of \cite{sj-zw} we use similar notation.
  By a black box operator we 
mean an operator $P$ defined on a domain ${\mathcal D}\subset \mch$ 
satisfying the conditions below.   
Let $R_0>0$ be fixed, and let $B(R_0)=\{x \in \Real^d: |x| <R_0\}$.
Let $\mch$ be a complex Hilbert space with orthogonal decomposition
$$\mch=\mch_{B(R_0)}\oplus L^2(\Real^d \setminus B(R_0)).$$
Using the notation of \cite{sj-zw}, we denote the corresponding orthogonal
projections by $u\mapsto u_{\restrict B(R_0)} $ and
$u \mapsto u_{ \restrict \Real^d  \setminus B(R_0)} .$  We assume that the 
operator $P:\mch\rightarrow \mch$ is semibounded below
and is self-adjoint with
domain  ${\mathcal D}\subset \mch$.  Furthermore, if  $u\in H^2(
\Real^d \setminus B(R_0))$ and 
$u$ vanishes near $B(R_0)$,  then $u\in {\mathcal D}$; and conversely
${\mathcal D}_{\restrict \Real^d \setminus B(R_0)}\subset H^2(\Real^d \setminus
B(R_0))$.
The operator $P$ is  $-\Delta $ outside $B(R_0)$: 
$$Pu_{\restrict  \Real^d \setminus B(R_0)} = -\Delta
u_{\restrict \Real^d  \setminus B(R_0)}
\; \mbox{for all }u\in { \mathcal D}$$
and $${\bf 1}_{B(R_0)}(P+i)^{-1} \; \mbox {is compact}$$
where ${\bf 1}_{B(R_0)}$ is the characteristic
function of $B(R_0)$.

We note that the Laplacian on $\Real^d \setminus \overline{\mco}$ with the boundary conditions
we have considered in the main part of the paper satisfy the black box conditions, with $\mch$ identified with $L^2(\Real^d \setminus \overline{\mco})$.

The proof of Proposition \ref{p:contat0}  will use \cite[Proposition 2.1]{pe-zwsc} which we
recall here for the convenience of the reader.  We have adapted the 
notation somewhat.
\begin{prop}(\cite[Proposition 2.1]{pe-zwsc})\label{p:pe-zw}   For $\phi\in C_c^{\infty}(\Real^d)$,
 let us denote by 
$${\mathbb E}^{\phi}_{\pm}(\lambda) :L^2(\Real^d)\rightarrow L^2(\Sphere^{d-1})$$
the operator with the kernel $\phi(x) \exp(\pm i \lambda \langle x,\omega \rangle )$.  Let us 
choose $\chi_i\in C_c^{\infty}(\Real^d)$, $i=1,\; 2,\;3$, such that $\chi_i\equiv 1$ near $U$
and $\chi_{i+1}\equiv 1$ on $\supp \chi_i$.  

Then for $0<\arg \lambda <\pi$ we have $S(\lambda)= I+A(\lambda)$, where 
$$A(\lambda) = i \pi (2\pi)^{-d} \lambda^{(d-1)/2} 
{\mathbb E}^{\chi _3}_+ (\lambda) [\Delta, \chi_1]
R(\lambda) [\Delta, \chi_2] ^t{\mathbb E}^{\chi_3}_-(\lambda)$$
where $^t{\mathbb E}$ denotes the transpose of ${\mathbb E}$.  The identity
holds for $\lambda \in \Lambda$ by analytic continuation.
\end{prop}

\begin{prop}\label{p:contat0}
Let the dimension $d$ be even, let $P$ be a black box compactly supported
perturbation of the Laplacian, and let $S(\lambda)$ be the corresponding 
scattering matrix, unitary on the positive real axis.  Then there is an 
$\epsilon >0$ so that $S(\lambda)$ is analytic in $V_{\epsilon}\defeq \{ 0\leq \arg \lambda 
\leq \pi,\; 0<|\lambda|<\epsilon\}$, and
$$\lim_{|\lambda|\rightarrow 0, \; \lambda \in V_{\epsilon}}S(\lambda)$$
exists.
\end{prop}
\begin{proof}  For $0<\arg \lambda <\pi$ we set $R(\lambda)= (P-\lambda^2)^{-1}$.
It is well known (see e.g. \cite{sj-zw}) that 
 $\chi \in C_c^{\infty}(\Real^d)$, 
$\chi R(\lambda)\chi$ has a meromorphic continution to $\Lambda$
and that $R(\lambda)$ 
has only finitely many poles in the region with $0\leq \arg \lambda \leq \pi$.  We are most
concerned here with a more delicate analysis near $0$, which through Proposition 
\ref{p:pe-zw} will give us information about the scattering matrix near $0$.

Let $\chi\in C_c^{\infty}(\Real^d)$.
By  \cite[Theorem 7.9]{mu-st}\footnote{
See also \cite{b-g-d, jensenspec, jensenspec4}, or \cite[Theorem 4.1]{murata}
which contain more detailed information for more specific cases.}
 given $M\in \Natural$ there are 
$ \beta_j \in \Natural_0$ and operators
$B_{j,k}$, so
that near $|\lambda|=0$, with $0\leq \arg \lambda \leq \pi$,
\begin{equation}\label{eq:hh}
\chi R(\lambda) \chi= \sum _{j=-2}^M \sum_{k=-\beta_j}^{\infty} 
\lambda^j (\log \lambda)^{-k} B_{j,k} + O(|\lambda|^{M-\delta})
\end{equation}
if $\delta>0$.
Morevoer, the coefficients of the terms which are unbounded at the origin have finite rank.
By Proposition \ref{p:pe-zw} the scattering matrix has a similar expansion near the origin.
In particular, there are at most a finite number of terms which are unbounded near $|\lambda|=0$, and each has finite rank.

Now we use the fact that for $\arg \lambda=0$, $|\lambda|>0$,
 $S(\lambda)$ is unitary so that $\|S(\lambda)\|=1$.  But this together with 
the expansion of $S(\lambda) $ near $|\lambda|=0$ means that the expansion cannot
have any terms which are unbounded as $|\lambda|\rightarrow 0$.
\end{proof}

It was observed in \cite{b-g-d} that dimension $d=2$
for scattering by a Schr\"odinger
operator $-\Delta +V$, $\lim_{\lambda \downarrow 0}S(\lambda)=I$  for any
 real-valued $V$ satisfying certain decay conditions.  This had 
earlier been noted for Schr\"odinger operators in dimension $d\geq 4$,
see \cite{jensenspec, jensenspec4}.
This contrasts with the case of dimensions $d=1$ and $d=3$, 
e.g. \cite{b-g-w,j-k}.  We show here
that a  similar phenomena holds in any even dimension
 for any operator $P$ 
satisfying the black-box conditions of Sj\"ostrand-Zworski, including
the exterior Laplacians of the type considered in the main body of the 
paper.  

We give below a proof of this which is somewhat algebraic, and hence
is rather different from the proof given in \cite{b-g-d} for $d=2$
for Schr\"odinger
operators.\footnote{ We note that \cite{b-g-d} proved, for 
Schr\"odinger operators, a much stronger
result than our Corollary \ref{cor:lat0}, since \cite{b-g-d} 
finds the first term or two in the asymptotic expansion of $S(\lambda)-I$
at the origin.}
\begin{cor}\label{cor:lat0}
Let $d$ be even, let
$P$ be any self-ajoint operator satisfying the black-box conditions of
Sj\"ostrand-Zworski recalled above, and let $S$ denote the corresponding
scattering matrix.  Then $\lim_{\lambda \downarrow 0} S(\lambda)=I$.
\end{cor}
\begin{proof}
By Proposition \ref{p:contat0} we can write $\lim_{\lambda \downarrow 0} S(\lambda)=S(0)
= \lim_{\lambda \downarrow 0} S(e^{i\pi } \lambda)$.

For $\lambda >0$, $S(\lambda)S^*(\lambda)=I$ and 
\begin{equation}\label{eq:relation}
S^*(\lambda)= 2I-\mcr S(e^{i\pi } \lambda)\mcr.
\end{equation}
  By continuity, both of these hold as well with $\lambda =0$.
In particular,  $S^*(0)S(0)=I$ and 
 any eigenvalue of $S(0)$ is of the form $e^{i\theta}$ for some 
$\theta \in \Real$.  Suppose $u$ is an eigenfunction of $S(0)$ with 
eigenvalue $e^{i\theta}$, and $\|u\|=1$.
Then
\begin{equation}\label{eq:applyident}
e^{i\theta}=\langle S(0) u,u\rangle =
\langle (2I-\mcr S^*(0) \mcr)u, u\rangle
\end{equation}
by (\ref{eq:relation}) at $\lambda=0$.  But since 
$\|\mcr S^*(0) \mcr\| \leq 1$, this means
$$|2-e^{i\theta}|= \left| \langle \mcr S^*(0) \mcr u, u\rangle \right|
\leq \|u\|^2 =1.
$$
Since we can have $|2-e^{i\theta}|\leq 1$ for $\theta\in \Real$ 
if and only if $e^{i\theta}=1$, we are done.
\end{proof}
We note that it is the application of (\ref{eq:relation}) in (\ref{eq:applyident}) that
is particular to the even-dimensional case.  The analogous relation
for the odd-dimensional case is that $S^*(e^{i\pi} \lambda)= \mcr S(\lambda) 
\mcr$ for $\lambda>0$, leading to the familiar
conclusion that if $e^{i\theta}$ is an eigenvalue of $S(0)$ in
odd dimensional black-box scattering, then $e^{i\theta}=\pm 1$.

\end{document}